\documentclass[a4paper,11pt]{article}

\usepackage{chngcntr}  % Include the chngcntr package

\usepackage[T1]{fontenc}

\usepackage[margin=1in]{geometry}
\usepackage{amsmath,amsthm,amssymb,accents}
\usepackage{nicefrac, xfrac}
\usepackage{bbm}
\usepackage{centernot} %negating symbols
\usepackage{tabularx}

\usepackage{pdflscape} %to change format to landscape

\usepackage{comment}
\usepackage{subcaption}
\usepackage{enumitem}
\usepackage{multirow,array}
\usepackage{mathtools}  
\usepackage{xfrac}
\usepackage{soul} %to highlight text
\usepackage{booktabs} %chatgpt wanted this for a table
\usepackage{tablefootnote} %for table footnotes to compile at the bottom of the page
\usepackage{epigraph,csquotes} %for fancy quotes
\usepackage{stackengine}  % provides \stackengine, \ensurestackMath, etc.
\usepackage{scalerel}     % provides \ThisStyle, \SavedStyle, \LMex, etc.

\usepackage{hyperref}
\hypersetup{
    pdftitle={Akkar: Dictatorial Committees}, % Title displayed in the thumbnail
    pdfauthor={D. Carlos Akkar}, % Optional: set the author metadata
    pdfsubject={Akkar Committees}, % Optional: subject
}

\usepackage{setspace}
\usepackage{needspace}

\usepackage{pgf,tikz,pgfplots}
\usetikzlibrary{patterns}
\pgfplotsset{compat=1.18}
\usepackage{mathrsfs}
\usetikzlibrary{arrows}

\usepackage[backend=biber,style=authoryear,sorting=nyt]{biblatex}
\AtEveryBibitem{%
  \iffieldundef{doi}{}{%
    \clearfield{url}%
    \clearfield{eprint}%
    \clearfield{urldate}%
  }%
}

\usepackage{graphicx}
\usepackage{capt-of} % For captions outside float environments
\numberwithin{figure}{section}

\usepackage[disable]{todonotes} %disables todonotes

\newcommand{\readernote}[1]{\todo[size=\scriptsize]{#1}}

\DeclareMathOperator{\Exp}{\mathbb{E}}
\DeclareMathOperator{\Prob}{\mathbb{P}}
\DeclareMathOperator*{\argmax}{arg\,max}
\DeclareMathOperator*{\argmin}{arg\,min}

\newcommand{\bigtimes}{%
  \mathop{\scalerel*{\times}{\bigcap}}\displaylimits
}

\usepackage{cleveref}
\usepackage{thmtools}
\usepackage{thm-restate}

\theoremstyle{definition}

\declaretheorem[name=Theorem]{thm}
\declaretheorem[name=Proposition, sibling=thm]{prop}
\declaretheorem[name=Lemma]{lem}

\declaretheorem[name=Claim]{claim}

\theoremstyle{definition}
\declaretheorem[name=Definition]{defn}

\numberwithin{equation}{section}

\crefname{thm}{Theorem}{Theorems}
\Crefname{thm}{Theorem}{Theorems}
\crefname{lem}{Lemma}{Lemmas}
\Crefname{lem}{Lemma}{Lemmas}
\crefname{prop}{Proposition}{Propositions}
\Crefname{prop}{Proposition}{Propositions}
\crefname{cor}{Corollary}{Corollaries}
\Crefname{cor}{Corollary}{Corollaries}
\crefname{claim}{Claim}{Claims}
\Crefname{claim}{Claim}{Claims}
\crefname{defn}{Definition}{Definitions}
\Crefname{defn}{Definition}{Definitions}
\crefname{cond}{Condition}{Conditions}
\Crefname{cond}{Condition}{Conditions}
\crefname{rem}{Remark}{Remarks}
\Crefname{rem}{Remark}{Remarks}

\newtheorem*{claim*}{Claim}

\newcommand{\commentout}[1]{}
\newcommand{\state}{\omega}

\newcommand{\supp}{\textrm{supp }}

\DeclareCiteCommand{\citeyear}
    {}
    {\bibhyperref{\printdate}}
    {\multicitedelim}
    {}

\newcommand\crownset[1]{\ensurestackMath{\ThisStyle{%
  \setbox0=\hbox{$\ThisStyle#1$}%
  \stackengine{.7pt}{\SavedStyle#1}{\kern\dimexpr-.1\LMex+.25\ht0\relax%
  \includegraphics[width=1.2\LMex]%
  {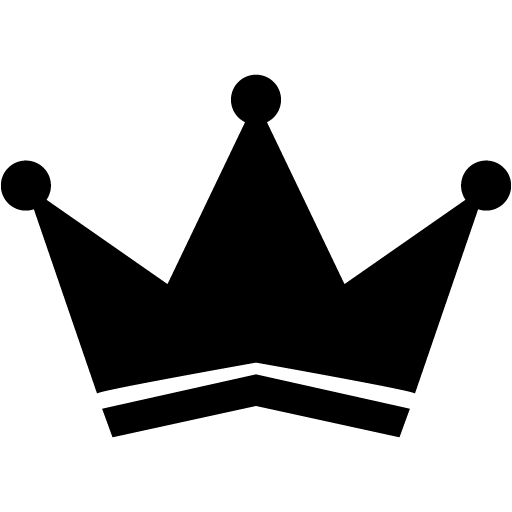}}%
  {O}{c}{F}{T}{S}}}}
\usepackage{mathrsfs}

\begin{document}

\vspace{15pt}

\title{Optimally Dictatorial Committees
\\
    }
\author{\href{https://www.carlosakkar.com}{D. Carlos Akkar}\footnote{Nuffield College and Department of Economics, Oxford. \href{mailto:akkarcarlos@gmail.com}{akkarcarlos@gmail.com}  \\
I thank César Barilla, Ian Jewitt, Ludmila Matysková, Margaret Meyer, Paula Onuchic, Manos Perdikakis, Daniel Quigley, and Ludvig Sinander for generous and helpful feedback.} 
}
\date{\today}
\maketitle
\readernote{COMMENTS ON!}

\vspace{-5pt}

\vspace{10pt}

\begin{abstract}
    \noindent I study the optimal voting mechanism for a committee that must decide whether to enact or block a policy of unknown benefit. Information can come both from committee members who can acquire it at cost, and a strategic lobbyist who wishes the policy to be enacted. I show that the \textit{dictatorship of the most-demanding member} is a \textit{dominant} voting mechanism: any other voting mechanism is (i) less likely to enact a good policy, (ii) more likely to enact a bad policy, and (iii) burdens each member with a greater cost of acquiring information. 
\end{abstract}

\newpage

\section{Introduction}
\label{section:intro}

\vspace{0.25cm}

\epigraph{\textit{[Democracy] is about making committees work.}}
         {Alan Bullock\footnotemark}
\footnotetext{\textcite{guardian_bullock}.}

\noindent
And a committee works only if it can harness information about the likely consequences of its decisions. This is a complex undertaking. The committee must first \textit{obtain} the information it needs. 
%This is a fight on two fronts: 
To do so, it must fight on two fronts: at home, it must motivate its members to \textit{acquire} information, even if this costs them time and effort. Outside, it must motivate external stakeholders to \textit{provide} information, even if this may harm their interests. The committee must then effectively \textit{aggregate} this information it has obtained: its decision must align with what this information warrants. How should such a  committee be designed, then, to best achieve these---possibly competing---objectives?

I study this question in a canonical and parsimonious setting.
There is a policy whose public benefit is unknown, and a committee wishes to decide whether to enact or block it. Committee members agree that a beneficial (\textit{good}) policy should be enacted, and that a harmful (\textit{bad}) one should be blocked. Nonetheless, each has different biases against the policy, so their preferred actions may differ while the policy's benefit remains uncertain. Outside the committee, there is a lobbyist. He prefers the policy to be enacted, whether it is good or bad. 

To influence the committee, the lobbyist can provide information about the policy's benefit. 
He does so by designing a public experiment about it, at no cost. Subsequently, committee members can acquire additional information about the policy's benefit. Each member does so privately, and at personal cost. Finally, committee members vote. They cast their votes simultaneously, and a pre-determined \textit{voting mechanism} prescribes (a) the set of votes each member may cast, and (b) the probability with which the policy will be enacted given the votes they cast.

I ask: is there a voting mechanism which secures
\vspace{-0.1cm}
\begin{enumerate}[itemsep=-1pt]
    \item the lowest expected cost of information acquisition for each member,
    \item lowest chance that a good policy is blocked,
    \item and the lowest chance that a bad policy is enacted
\end{enumerate}
\vspace{-0.1cm}
among \textit{all} possible voting mechanisms? Briefly put, is there a \textit{dominant} voting mechanism?

In Theorem \ref{thm:optimaldictatorship}, I prove that a dominant voting mechanism does indeed exist by characterising one: \textit{the dictatorship of the most-demanding member}. In it, only one member's vote influences the outcome. She is the \textit{most-demanding member}: the one who would obtain the most certainty that the policy is good before enacting it, were she deciding without a committee or lobbyist to interfere. 

There are a myriad of complex mechanisms. Why does this simple one dominate them all?

\vspace{0.12cm}
\noindent 1) In a dictatorship, no member expects to acquire any information.
\vspace{0.12cm}

Clearly, the dictator's peers have no incentive to acquire any information---their vote never affects the outcome. The dictator, in contrast, would wish to acquire information unless she were sufficiently certain about the policy's unknown benefit. Crucially, the lobbyist always provides enough information to deliver that certainty. He sees no harm in doing so: if he did not, the dictator would simply acquire that information herself. He reaps ample benefit from this, too---I discuss this next. Thus, the dictator receives all her information from the lobbyist, and simply votes accordingly.

\vspace{0.12cm}
\noindent 2) In a dictatorship, a good policy is never blocked.
\vspace{0.12cm}

A dictatorship incentivises the lobbyist to be the dictator's sole source of information---he provides enough information for the dictator to forgo acquiring any herself. The lobbyist leverages this position: he provides any additional information that raises the probability with which the policy is enacted. 
As such, he ensures that the dictator never blocks the policy unless she is certain that it is bad. Otherwise, he could provide her additional information that would either merely reinforce the dictator's belief that the policy should be blocked, or convince her that the policy is sufficiently likely to be good and turn her verdict around.

\vspace{0.12cm}
\noindent 3) In the dictatorship of the most-demanding member, a bad policy is least likely to be enacted.

To enact the policy without acquiring further information, a dictator must be sufficiently certain that it is good. That threshold level of certainty is highest for the most-demanding member: at her threshold, any dictator would enact the policy. I show that, in fact, \textit{any} committee would enact the policy at this threshold---no committee is more demanding than the most-demanding member. 
Intuitively, a committee member values information the most when she is a dictator; other voting mechanisms might anyway fail to incorporate her information into the final outcome. So, at a belief threshold where no dictator would acquire information before enacting the policy, no committee member would do so either before voting to enact it. 

As no committee is more demanding than the most-demanding member, the lobbyist can ensure that no committee is less likely to enact the policy than the most-demanding member's dictatorship. But the most-demanding member's dictatorship always enacts good policies; so, I conclude that no committee is less likely to enact a bad policy than the most-demanding member's dictatorship. 

\subsection*{Contribution and Literature}

I study how a committee must be designed in the face of its fourfold informational problem:
\vspace{-0.1cm}
\begin{enumerate}[itemsep=-1pt]
    \item it must incentivise external stakeholders to provide information relevant to its decision,
    \item it must incentivise its members to acquire such information, even if they find this costly,
    \item it must minimise the cost its members bear as they acquire such information,
    \item its final decision must efficiently aggregate the information it has accumulated.
\end{enumerate}

There is an extensive and multidisciplinary body of literature around this fourfold informational problem. However, existing literature addresses only a subset of these four problems at a time, ignoring the rest.
%rather than their simultaneous interplay. 
%
An earlier strand studies environments without external stakeholders to interfere with the committee's decision. This strand asks how a committee may aggregate its members' private information---both when information is held exogenously\footnote{
    See \textcite{condorcet_jurytheorem}'s celebrated Jury Theorem, and the seminal works of \textcite{austen-smith_banks_1996} and \textcite{feddersen_pesendorfer_96}, (
    \citeyear{feddersen_pesendorfer_1997}, and \citeyear{Feddersen_Pesendorfer_1998}.
}
and when it is acquired at personal cost\footnote{
    See \textcite{gersbach_1995}, \textcite{li_2001_conservatism}, \textcite{persico_2004_restud}, \textcite{smorodinsky_tennenholtz_2004} and (\citeyear{smorodinsky_tennenholtz_2006}), \textcite{gerardi_yariv_2008}, \textcite{gershkov_szentes_2009}, and \textcite{chan_lizzeri_suen_yariv_2018}. \textcite{stephenson_2011} offers a legal-theoretic discussion and treatment.
}.
A more recent strand, in turn, focuses on external stakeholders who may try to influence the committee through providing information, but mutes committee members' ability to acquire additional information.\footnote{
    See \textcite{breyer_82} and \textcite{coglianese_2004} for legal-theoretic discussions. See \textcite{alonso_camara_2016}, \textcite{bardhi_guo_2018}, \textcite{chanetal_2019}, and \textcite{heese_lauermann_2024} for contributions to Economics.}

However, the objectives these four problems impose on the committee are often in conflict. Consequently, so is the guidance offered by these two strands of the literature on optimal committee design.\footnote{
    For instance, both \textcite{persico_2004_restud} and \textcite{alonso_camara_2016} study the optimal $k$-majority rule in the (suitably adapted) canonical voting setting this paper examines, preceded by \textcite{austen-smith_banks_1996}. In \textcite{persico_2004_restud}, there are no external stakeholders to inform the committee, unlike \textcite{alonso_camara_2016}. In contrast, \textcite{alonso_camara_2016}'s voters (replacing the \textit{committee members} here) have no access to private information, unlike \textcite{persico_2004_restud}. Consequently, \textcite{persico_2004_restud} finds that large majority rules are often suboptimal, as they disincentivise committee members from acquiring private information. In contrast, \textcite{alonso_camara_2016} find that a unanimity rule is optimal when voters' preferences are sufficiently aligned, as this rule maximises the politician's (replacing \textit{the lobbyist} here) incentives to provide information to the voters. 
}
This is troubling for policymakers who may want to apply these lessons, since policymaking committees often face all four of these problems simultaneously.\footnote{
    \textcite{breyer_82} highlights this for agencies regulating public utilities: ``[An agency] will obtain the information, in part, through research by agency staff, as they consult research literature and talk to employees of other agencies. ... Once the Notice of Proposed Rulemaking is promulgated, however, ... information is then more likely to flow in the form of comments submitted to the agency by lawyers for parties.}

This paper addresses this gap by studying the optimal voting mechanism for a committee facing all four of these problems simultaneously. 
In Theorem \ref{thm:optimaldictatorship}, I find that a dictatorship {dominates} any other voting mechanism. This is robust to committee members' information acquisition costs and biases towards the policy; these characteristics only determine \textit{who} the dictator should be. This robustness contrasts with the earlier literature on information acquisition in committees, which generally finds that the optimal mechanism may depend greatly on committee members' characteristics.\footnote{
    For instance, \textcite{persico_2004_restud} finds that the optimal majority rule depends on the precision of the information available to committee members. \textcite{gerardi_yariv_2008} find that the optimal size of a committee depends on its members' cost of acquiring information, and how precise that information is. \textcite{gershkov_szentes_2009} find that the whether the optimal sequential voting rule is ex-post efficient similarly depends on agents' (replacing the \textit{committee members} here) cost of acquiring information. 
}

Among papers in the aforementioned literature, Theorem \ref{thm:optimaldictatorship} aligns most closely with \textcite{alonso_camara_2016} and \textcite{bardhi_guo_2018}. In both, a politician, or Sender (replacing \textit{the lobbyist} here), 
provides information to persuade a committee whose members have no access to private information. Both find that, when members' preferences align to the extent they do in this paper,\footnote{
    Both consider a more general class of preferences for members than this paper. \textcite{alonso_camara_2016} find this result when members agree on which states make enacting the policy more desirable, and in which states it should be enacted. \textcite{bardhi_guo_2018} find it when members' preference states are perfectly correlated. 
} each member prefers the unanimity rule to any other $k$-majority rule.\footnote{
    \textcite{bardhi_guo_2018} study both the case where information provision is \textit{general}---the Sender publicly communicates all information---and the case where it is \textit{private}---the Sender can provide different information to different members. They find this for the former, but not necessarily the latter case.
}
This is because under unanimity, the \textit{most-demanding member}---there, the one with the strongest bias against the policy---becomes the pivotal voter. 
The politician, or Sender, must then persuade her to enact the policy. This compels him to provide more information to the committee than he would under any other $k$-majority rule. 

I also find that under the dictatorship of the most-demanding member, the lobbyist is incentivised to provide more information to the committee than under any other voting mechanism (Proposition \ref{prop:dict_mostinfo}). However, this only provides a partial intuition for Theorem \ref{thm:optimaldictatorship}:
even if such a voting mechanism extracts less information from the lobbyist, it may still provide strong information acquisition incentives to committee members. As such, such a voting mechanism may still lead to better outcomes than the dictatorship of the most-demanding member---in principle. 

What drives Theorem \ref{thm:optimaldictatorship} is that the lobbyist \textit{always} prevents a voting mechanism from outperforming the dictatorship of the most-demanding member as such. When the most-demanding member is the dictator, the information she receives from the lobbyist dissuades her from acquiring more herself. In proving Theorem \ref{thm:optimaldictatorship}, I establish that the same amount of information dissuades \textit{any} committee member from acquiring information, regardless of the voting mechanism. I then show that the lobbyist provides less than this amount of information only when this persuades the committee to enact bad policies more often---potentially, even at the cost of enacting good policies less often.

\textcite{matyskova_montes_jet_2023} study a model akin to mine, but where the Sender (here, the \textit{lobbyist}) attempts to persuade a single decision-maker (here, a \textit{committee member}) rather than a committee. Their main result finds that it is without loss to restrict to equilibria where the Sender gives the Receiver enough information to dissuade her from acquiring any herself. My Lemma \ref{lem:dictatoreqm}, which finds that the lobbyist always provides enough information to disincentivise the dictator from acquiring any herself, can be viewed as a manifestation of their main result. 

Lastly, it is worth remarking that this paper only studies \textit{simultaneous} voting mechanisms. Some authors\footnote{
    For instance, \textcite{gershkov_szentes_2009}, and \textcite{chan_lizzeri_suen_yariv_2018}.
}
have also studied what the optimal \textit{sequential} voting mechanism is for a committee. Investigating this in settings where the committee must incentivise both information acquisition by its members and information provision by a strategic lobbyist remains a promising area of research.

\subsection*{Outline}

The remainder of the paper is organised as follows. Section \ref{section:model} introduces the model. Section \ref{section:dictatorships} introduces dictatorships and presents Lemma \ref{lem:dictatoreqm}, which discusses the equilibria dictatorships induce. Section \ref{section:optimality} presents and proves Theorem \ref{thm:optimaldictatorship}, the main result. It also presents Proposition \ref{prop:dict_mostinfo}, which proves that no voting mechanism incentivises the lobbyist to provide more information than he does under the dictatorship of the most-demanding member.
Section \ref{section:discussion} concludes by discussing 
how sensitive the results of the paper are to certain modelling assumptions.

\section{The Model}
\label{section:model}

A committee of $N$ \textit{members}, $m_1, ..., m_N$, wish to decide whether to \textit{enact} or \textit{block} a policy. 
The benefit member $m_i$ obtains from enacting the policy depends on an unknown state of the world, $\state \in \left\{  \state_0, \state_1 \right\}$. When $\state = \state_1$---the policy is \textit{good}---she obtains a payoff of $u_i \geq 0$ from enacting it. When $\state = \state_0$---the policy is \textit{bad}---she incurs a payoff loss of $1$ instead. Blocking the policy, in contrast, yields her a payoff of $0$ in either state. There is also a \textit{lobbyist} outside the committee. He prefers the policy to be enacted regardless of the state. 

\sloppy At the outset of the game, the lobbyist and committee members share the prior belief that $\state = \state_1$ with probability $p \in (0,1)$. The lobbyist can then provide additional information about the state. He does so by committing to perform and publicly disclose the outcome of an (Blackwell) experiment of his choosing. Formally, an experiment is comprised of a measurable space  $\left( X, \mathcal{X} \right)$ and, for each possible state $\state_k$, a probability measure $P_k$ over this space. Upon observing the outcome of the lobbyist's experiment, players use Bayes Rule to update their prior belief $p$ to an \textit{interim belief} $r \in [0,1]$.

I follow a \textit{belief-based approach}\footnote{
    See, for instance, \textcite{kg2011}, and \textcite{CaplinDean2013}.
} and model the lobbyist as directly choosing the probability measure $\mu_l$ his experiment will induce over the space of interim beliefs, $[0,1]$. He can choose any such measure, provided it is \textit{Bayes plausible}; i.e. $\int\limits_{[0,1]} r \mu_l (d r) = p$. I denote the set of such measures as $\Delta_p ( [0,1] )$.

After the lobbyist reveals the outcome of his experiment, each committee member may acquire additional information about the state. She may do so by choosing an additional experiment whose outcome she observes. Information acquisition by committee members is \textit{private}; neither $m_i$'s chosen experiment nor its outcome are disclosed to anyone.

Following the outcome of her experiment, $m_i$ updates her interim belief $r$ to a \textit{posterior belief} $q_i \in [0,1]$ that $\state = \state_1$. Again, I model each member $m_i$ as choosing the probability measure $\mu_{i} ( . \mid r) \in \Delta_r ([0,1])$ her experiment will induce over her posterior beliefs, given her interim belief $r$. 
To simplify exposition, I assume that the outcomes of members' chosen experiments are independent conditional on the true state $\state$.

Producing and providing public information is costless for the lobbyist. In contrast, acquiring private information is costly for committee members: $m_i$ must pay a cost $C_i (\mu)$ to observe the outcome of an experiment which induces the distribution $\mu$ over her possible posterior beliefs. The cost function $C_i$ maps $\Delta ( [0,1] )$, the set of all probability measures over $[0,1]$, to $\mathbb{R}^+$. It is \textit{uniformly posterior separable};\footnote{See, for instance, \textcite{caplin_dean_leahy_2017}, and \textcite{denti_2022}.} i.e., there is a convex function $c_i: [0,1] \to \mathbb{R}^+$ such that:
\begin{equation*}
    C_i (\mu) = \int\limits_{[0,1]} c_i ( q ) d \mu(q) - c_i \left( \Exp_{\mu} [q] \right)
\end{equation*}

Finally, each committee member casts a private vote into a \textit{voting mechanism}. A voting mechanism is a tuple $\left( D, \left\{ V_i \right\}_{i=1}^{N}, \left\{ \mathcal{V}_i \right\}_{i=1}^{N} \right)$ where $V_i \ni v_i$ is a set of possible votes each member $m_i$ may cast, $\mathcal{V}_i$ is a $\sigma$-algebra over $V_i$, and the \textit{decision rule} $D: \bigtimes\limits_{i=1}^{N} V_i \to [0,1]$ is a measurable mapping from the space of possible vote profiles to the probability that the policy is enacted. I require that:
\begin{enumerate}[itemsep=-5pt]
    \item given others' votes, member $m_i$ have a \textit{lowest} and \textit{highest} vote to submit:\footnote{
        This requirement is automatically satisfied when the set of possible vote profiles $\bigtimes\limits_{i=1}^{N} V_i$ is finite, or when $\bigtimes\limits_{i=1}^{N} V_i$ is compact and the decision rule $D(\cdot)$ is continuous. 
    }
    \vspace{-0.25cm}
    \begin{equation*}
        \argmin\limits_{v_i \in V_i} D (v_i, \mathbf{v_{-i}}) \text{ and } \argmax\limits_{v_i \in V_i} D (v_i, \mathbf{v_{-i}}) \text{ exist for each } \mathbf{v_{-i}} \in \mathbf{V_{-i}} := \bigtimes\limits_{j \neq i} {V_j}
    \end{equation*}
    \vspace{-0.5cm}
    \item there be a profile of votes upon which the policy is almost surely enacted (blocked):
    \vspace{-0.25cm}
    \begin{equation*}
        \exists \text{ } \mathbf{\underaccent{\bar}{v}}, \mathbf{\Bar{v}} \in \bigtimes\limits_{i=1}^{N} V_i \text{ such that } D ( \mathbf{\underaccent{\bar}{v}}  ) = 0 \text{ and } D ( \Bar{\mathbf{v}} ) = 1
    \end{equation*}
\end{enumerate}

For any given voting mechanism, I focus on the lobbyist-preferred perfect Bayesian equilibrium of this game. A perfect Bayesian equilibrium consists of:
\begin{enumerate}
    \item a voting strategy $\left\{ \tau_{i}^* (. \mid q,r) \right\}_{(q,r) \in [0,1]^2}$ for each member $m_i$, where $\tau_{i}^* (. \mid q,r) \in \Delta ( V_i )$,
    \item an information strategy $\left\{ \mu_{i}^* (. \mid r) \right\}_{r \in [0,1]}$ for each member $m_i$, where $\mu_{i}^* (. \mid r) \in \Delta_{r} \left(  [0,1] \right)$,
    \item a strategy $\mu_l^* \in \Delta_p \left( [0,1] \right)$ for the lobbyist,
\end{enumerate}
such that:
\begin{enumerate}%[label=(\arabic*)]
    \item Given other members' strategies, member $m_i$'s voting strategy
    %behaviour $\tau^*_{i; q,r}$ 
    maximises her expected payoff at any pair of interim and posterior beliefs $(r,q) \in [0,1]^2$ she may hold:
    \begin{equation*}
        \begin{split}
            &v_i \in \operatorname*{supp } \tau_{i}^* ( . \mid q,r ) \\
            \implies &v_i \in \operatorname*{argmax}_{v \in V_i} \int\limits_{[0,1]^{N-1}} U_i ( r, q_i, \mathbf{q_{-i}} ) \times \left\{ \int\limits_{V_{-i}} D (v, \mathbf{v_{-i}}) \tau^*_{-i} (d \mathbf{v_{-i}} \mid r, \mathbf{q_{-i}})  \right\} \mathbf{\mu_{-i}^*} \left( d \mathbf{q_{-i}} \mid r, q_i \right)
        \end{split}
    \end{equation*}
    where: 
    \begin{enumerate}[label=\roman*.]
        \item $\mathbf{\mu_{-i}^*} \left( \cdot \mid r, q_i \right)$ is the probability measure $m_i$ assigns over the posterior beliefs of all members excluding herself, conditional on her interim and posterior beliefs $r$ and $q_i$, and those members' information strategies.\footnote{
            This measure can easily be derived using Bayes' Rule.
        }
        \item $\tau^*_{-i} ( \cdot \mid r,\mathbf{q}_{-i}) := \bigtimes\limits_{j \neq i} \tau^*_{j} (. \mid r, q_j) $ is the probability measure 
        over the votes of all members excluding $m_i$, conditional on their interim and posterior beliefs $r$ and $\mathbf{q}_{-i}$, and their voting strategies.
        \item $U_i ( r, q_i, \mathbf{q_{-i}} )$ denotes the payoff $m_i$ expects from enacting the policy given all players' information about the state:
        \begin{equation*}
            U_i ( r, q_i, \mathbf{q_{-i}} ):=  \Prob_{r, q_1, \dots q_N} (\state_1) \times (u_i + 1) -1
        \end{equation*}
        Here, $\Prob_{r, q_1, \dots q_N} (\state_1)$ denotes the probability that $\state = \state_1$ given the same information:
        \[
            \Prob_{r,q_1,...,q_N} (\state_1) := \left( \frac{r}{1-r} \times \prod\limits_{j=1}^{N} \frac{q_j/(1-q_j)}{r/(1-r)} \right) \times \left( 1 + \frac{r}{1-r} \times \prod\limits_{j=1}^{N} \frac{q_j/(1-q_j)}{r/(1-r)} \right)^{-1}
        \]

        \end{enumerate}

    \item Given others' strategies, member $m_i$'s information strategy maximises her expected payoff at any interim belief $r$:
    \begin{equation*}
        \mu_i^*( \cdot \mid r ) \in \argmax\limits_{\mu \in \Delta_r ([0,1])} \int \pi_i^* ( r, q ) \mu ( d q ) - C_i (\mu)
    \end{equation*}
    where $\pi_i^* ( r, q )$ is the payoff $m_i$ expects in the ensuing subgame, conditional on her interim belief $r$, posterior belief $q$, and ensuing voting strategy,
    \item The lobbyist's information strategy $\mu_l^*$ maximises the probability that the policy is enacted given committee members' information and voting strategies.
    \end{enumerate}
    
    The lobbyist-preferred perfect Bayesian equilibrium (hereon, simply \textit{equilibrium}) is the perfect Bayesian equilibrium where the probability that the committee enacts the policy is highest. 

\subsection*{Dominant Voting Mechanisms}

A voting mechanism is \textit{dominant} if it gives rise to an equilibrium that yields:
\begin{enumerate}[label=\roman*.]
    \item a (weakly) higher probability that the policy is blocked when $\state = \state_0$,
    \item a (weakly) higher probability that the policy is enacted when $\state = \state_1$, and
    \item a (weakly) lower expected cost of information acquisition for each committee member
\end{enumerate}
compared to any equilibrium that any other voting mechanism yields. Every committee member (weakly) prefers a dominant voting mechanism over any other voting mechanism.

\section{Dictatorships}
\label{section:dictatorships}

One of the simplest voting mechanisms the committee may adopt is a \textit{dictatorship}. In a dictatorship, only one member's vote influences the outcome. 

\begin{defn}
    A voting mechanism $( \crownset{D_i}, \left\{ {V}_j \right\}_{j=1}^{N}, \left\{ \mathcal{V}_j \right\}_{j=1}^{N} )$ is an $m_i$\textit{-dictatorship} if:
    \vspace{-0.25cm}
    \[
        \crownset{D_i}(v_i,\mathbf v_{-i})
            = \crownset{D_i}(v_i,\mathbf v'_{-i})
            \quad \text{for any }
            v_i\in V_i \text{ and all } \mathbf v_{-i},\mathbf v'_{-i}\in \mathbf{V_{-i}}
    \]
\end{defn}

Lemma \ref{lem:dictatoreqm} characterises the equilibria dictatorships induce. 

\begin{lem}
    \label{lem:dictatoreqm}
    In the equilibrium of an $m_i$-dictatorship:
    \begin{enumerate}
        \item $m_i$'s voting strategy satisfies:
        \[
        \Exp_{ \tau_i^* ( . \mid r,q ) } \crownset{D_i} = 
            \begin{cases}
                0 & \text{for } q < ( 1 + u_i )^{-1} 
                \\
                1 & \text{otherwise}
            \end{cases}
        \]
        for any interim belief $r \in [0,1]$.
        \item $m_i$'s information acquisition strategy $\left\{ \mu_i^* (. \mid r) \right\}_{r \in [0,1]}$ is the unique one such that for every $r \in [0,1]$, $\mu_{i}^* ( \cdot \mid r) \in \Delta_r ([0,1])$, and:
        \vspace{-0.25cm}
        \[
            \supp \mu_{i}^* ( \cdot \mid r) = 
            \begin{cases}
                \left\{ \underaccent{\bar}{q}_i^*, \Bar{q_i}^* \right\} & r \in \left( \underaccent{\bar}{q}_i^*, \Bar{q_i}^* \right) \\
                \{ r \} & \text{otherwise}
            \end{cases}
        \]
        %\vspace{-0.25cm}
        %
        where $\underaccent{\bar}{q}_i^*$ and $\Bar{q_i}^*$ are two threshold beliefs such that
        $\underaccent{\bar}{q}_i^* \leq (1 + u_i)^{-1} \leq  \Bar{q_i}^*$.
        \item The lobbyist's information strategy $\mu_l^*$ is the unique probability measure satisfying:
        \vspace{-0.25cm}
        \[
            \supp \mu_l^* =
            \begin{cases}
                \left\{ 0, \Bar{q_i}^* \right\} & p < \Bar{q_i}^* \\
                \{ p \} & \text{otherwise}
            \end{cases}
        \]
        \item Members except $m_i$ never incur a positive information acquisition cost; i.e., $C_j \left( \mu_j^* (. \mid r) \right) = 0 \quad \text{for every } r \in [0,1]$.
    \end{enumerate}
\end{lem}

In an $m_i$-dictatorship, $m_i$ has exclusive control over the decision.
Thus, she votes to ensure that the decision maximises her expected payoff: the policy is enacted whenever she expects a weakly positive payoff from doing so, and it is blocked otherwise. 

$m_i$'s information strategy is fully characterised by her two \textit{thresholds of persuasion}, $\underaccent{\bar}{q}_i^*$ and $\Bar{q_i^*}$. $m_i$ acquires additional information about the state if and only if her interim belief $r$ falls between these thresholds. When she does, the information she acquires moves her posterior belief precisely to one of these two thresholds. Importantly, $m_i$'s thresholds of persuasion always bracket her \textit{threshold of indifference} $(1+u_i)^{-1}$; i.e., the belief which leaves her indifferent between enacting and blocking the policy. This reflects that $m_i$ never acquires a costly experiment whose outcome always leads to the same vote, and as such, has no value for her decision. In contrast, $m_i$ acquires no further information when her interim belief lies beyond these thresholds, $r \notin ( \underaccent{\bar}{q}_i^*, \bar{q_i}^* )$. Instead, she simply takes whatever decision her interim belief favours. Her peers, having no influence over the outcome, never acquire costly information about the state.

In turn, the lobbyist wishes to maximise the probability that the committee enacts the policy. In an $m_i$-dictatorship, this amounts to maximising the probability that $m_i$'s posterior belief weakly exceeds her threshold of indifference, $(1+u_i)^{-1}$, so that she votes for enactment. Owing to $m_i$'s simple information strategy, the lobbyist can restrict himself to interim beliefs beyond $m_i$'s thresholds of persuasion, $r \notin \left( \underaccent{\bar}{q}_i^*, \Bar{q_i^*} \right)$. Interim beliefs between these thresholds merely prompt $m_i$ to distribute her posterior beliefs back to her thresholds of persuasion. In contrast, those outside that interval prompt her to vote without acquiring further information.\footnote{
    \textcite{matyskova_montes_jet_2023} establish a more general version of this argument. A Sender (replacing \textit{the lobbyist} here) who faces a single Receiver (replacing \textit{the dictator} here) with a uniformly posterior separable information acquisition cost suffers no loss from restricting himself to interim beliefs where the Receiver acquires no further information.}

This reduces the lobbyist's problem to a simple Bayesian Persuasion exercise a là \textcite{kg2011}. When $m_i$'s prior belief already exceeds her upper threshold of persuasion $\Bar{q_i^*}$, the lobbyist provides no further information. Then, $m_i$ simply enacts the policy without acquiring any additional information. Otherwise, the lobbyist's experiment either induces the interim belief $0$, or $\Bar{q_i^*}$. The latter precisely equals $m_i$'s upper threshold of persuasion: $m_i$ is persuaded to enact the policy, but with no more evidence than necessary. The former, on the other hand, conclusively reveals that the policy is bad, persuading $m_i$ to block the policy. Only a bad policy can induce this interim belief, ensuring that $m_i$ never blocks a good policy. 

\needspace{3\baselineskip} % Checks for enough space for about 5 lines
\begin{proof}[Proof, Lemma 1:]
    \nopagebreak
    \leavevmode\vspace{0.15em}
    \nopagebreak
    \begin{enumerate}
        \item \sloppy Since only $m_i$''s vote affects the outcome, we can define the measurable function $\crownset{f_i} : V_i \to [0,1]$ to be $\crownset{f_i} (v_i) := \crownset{D_i} (v_i, \mathbf{v_{-i}})$ for an arbitrary $\mathbf{v_{-i}} \in V_{-i}$. Then, $v_i \in \supp \tau_i^* ( \cdot \mid r,q )$ only if:
        \begin{align*}
                &v_i \in \argmax_{v \in V_i} \int\limits_{[0,1]^{N-1}} U_i (r, q_i, \mathbf{q_{-i}}) \times \left\{ \int\limits_{V_{-i}} \crownset{f_i} (v_i) \tau^*_{-i} (d \mathbf{v_{-i}} \mid r, \mathbf{q_{-i}})  \right\} \mathbf{\mu_{-i}^*} \left( d \mathbf{q_{-i}} \mid r, q_i \right)
                \\
                \implies & v_i \in \argmax_{v \in V_i} \crownset{f_i} (v_i) \times \int\limits_{[0,1]^{N-1}} U_i (r, q_i, \mathbf{q_{-i}}) \mathbf{\mu_{-i}^*} \left( d \mathbf{q_{-i}} \mid r, q_i \right)
                \\
                \implies & v_i \in \argmax_{v \in V_i} \crownset{f_i} (v_i) \times \left[ q_i \times ( 1 + u_i ) -1 \right]
        \end{align*}
        where the last implication follows from the martingale property of beliefs; i.e., $\Exp_{\mathbf{q_{-i}}} [\Prob_{r,q_1,...,q_N} \mid q_i ] = q_i$. 
        
        Thus, $ \{ \tau_i^*( \cdot \mid r,q ) \}_{(r,q) \in [0,1]^2} $ is an equilibrium strategy only if:
        \begin{equation*}
            \Exp_{\tau_i^* ( \cdot \mid r,q )} \crownset{f_i} (.) =
            \begin{cases}
                0 & q < (1 + u_i)^{-1} \\
                1 & q > (1 + u_i)^{-1}
            \end{cases}
        \end{equation*}
        When $q = (1+u_i)^{-1}$, $m_i$ is indifferent between enacting and blocking the policy, hence any vote. He must then submit the lobbyist-preferred vote; i.e., $\Exp_{\tau_i^* ( \cdot \mid r, (1 + u_i)^{-1} )} \crownset{f_i} (.) = 1$.

        \item  $m_i$ only needs to decide whether to vote to maximise or minimise the enactment probability of the policy. So, she can restrict to information strategies where $\lvert \supp \mu_{i}^* ( \cdot \mid r )  \rvert \leq 2$, without loss of generality.
        Furthermore, $m_i$ can always restrict herself to either acquiring no information, or acquiring information that has decision-making value; i.e., $(1 + u_i)^{-1} \in \text{conv} \left( \supp \mu_{i}^* ( \cdot \mid r ) \right)$ whenever $\lvert \supp \mu_{i}^* ( \cdot \mid r )  \rvert = 2$. Lastly, note that Theorem 1 in \textcite{caplin_dean_leahy_2022} implies that $m_i$'s strategy must be \textit{Locally Posterior Invariant} when $C_i$ is uniformly posterior separable; i.e., for any $\Tilde{r} \in \text{conv} \left( \supp \mu_{i}^* ( \cdot \mid r ) \right)$, $\mu^*_{i} ( \cdot \mid \Tilde{r} )$ is the unique distribution such that $\supp \mu_{i}^* ( \cdot \mid \Tilde{r} ) = \supp \mu_{i}^* (\cdot \mid r)$.

    Letting $ \left\{ \underaccent{\bar}{q}_{i}^*, \Bar{q}_i^* \right\} := \supp \mu_{i}^* (\cdot \mid (1 + u_i)^{-1})$ and $\underaccent{\bar}{q}_i^* \leq \Bar{q}_i^*$, we thus conclude:
    \begin{enumerate}
        \item for all $r \in \left[ \underaccent{\bar}{q}_{i}^*, \Bar{q}_i^* \right]$, $\mu_{i}^* (\cdot \mid r)$ is the unique distribution such that $\text{supp} \mu_{i}^* (\cdot \mid r) = \left\{ \underaccent{\bar}{q}_{i}^*, \Bar{q}_i^* \right\}$.
        \item for all $r \notin \left[ \underaccent{\bar}{q}_{i}^*, \Bar{q}_i^* \right]$, $\mu_{i;r}^* = \delta_r$. 
    \end{enumerate}
    \item Given members' voting and information strategies, let $\zeta^*_i (r)$ denote the probability that the policy is enacted in an $m_i$-dictatorship, following an interim belief of $r$. Denote the concave envelope of this function as $\hat{\zeta_i}^*(r)$. These two functions are given by:
    \begin{align*}
        \zeta^*_i (r) &:= 
        \begin{cases}
            0 & r \leq \underaccent{\bar}{q}_i^* \\
            \frac{ r - \underaccent{\bar}{q}_i^* }{ \Bar{q}_i^* - \underaccent{\bar}{q}_i^* } & r \in \left[ \underaccent{\bar}{q}_i^*, \Bar{q}_i^* \right] \\
            1 & r \geq \Bar{q}_i^*
        \end{cases}
        &
        \hat{\zeta_i}^*(r) &:= 
        \begin{cases}
            \frac{r}{\Bar{q}_i^*} & r \leq \Bar{q}_i^* \\
            1 & r \geq \Bar{q}_i^*
        \end{cases}
    \end{align*}
    %
    %So, $\mu_l^* \in \operatorname*{argmax}\limits_{ \mu \in \Delta_{p_0} \left( [0,1] \right) } \Exp_{\mu} \pi_{l;i}^*(r) $. 

    That $\hat{\zeta_{i}}^*(r)$ is indeed the concave envelope of $\zeta_{i}^* (r)$ can be easily verified: any concave function that coincides with $\zeta_{i}^*$ at $r = 0$ and $r = \Bar{q}_i^*$ must weakly lie above $\hat{\zeta_{i}}^*(r)$ on the interval $[0, \Bar{q}_i^*]$.

    The condition $\Exp_{\mu} \zeta_{i}^* (.) = \hat{ \zeta_{i} }^* (p)$
    is then sufficient for the optimality of $\mu$ for the lobbyist (\textcite{kg2011}).
    It is easily seen that the distribution $\mu_l^*$ defined in Lemma \ref{lem:dictatoreqm} satisfies this condition. 
    
    \end{enumerate}
    
\end{proof}

\section{A Dominant Dictatorship}
\label{section:optimality}

The equilibrium of an $m_i$-dictatorship is shaped by the dictator's upper threshold of persuasion, $\Bar{q_i}^*$. In this equilibrium, no committee member acquires any information; the committee relies entirely on the information provided by the lobbyist. In turn, the lobbyist ensures that the committee blocks the policy only when it is conclusively revealed to be bad. In contrast, he never \textit{over-persuades} $m_i$, the dictator, to enact the policy; $m_i$'s interim belief never \textit{strictly} exceeds her upper threshold of persuasion, $\Bar{q_i}^*$, unless her prior belief already does. 

So, the higher a dictator's upper threshold of persuasion is, the more information she extracts from the lobbyist. I call the member for whom this threshold is highest the \textit{most-demanding member} of the committee. 

\begin{defn}
    Member $m_i$ is the \textit{most-demanding member} of the committee, denoted $\crownset{m}$, if $\Bar{q}_{i}^* = \max \left\{ \Bar{q}_1^*, \Bar{q}_2^*, ..., \Bar{q}_N^* \right\} =: {\crownset{q}}$. 
\end{defn}

The dictatorship of the most-demanding member clearly dominates the dictatorship of any other member---no dictatorship results in costly information acquisition or a good policy getting blocked, but the $\crownset{m}$-dictatorship minimises the probability of a bad policy getting enacted. Theorem \ref{thm:optimaldictatorship} shows that, in fact, the $\crownset{m}$-dictatorship dominates \textit{every} voting mechanism.

\begin{thm}
    \label{thm:optimaldictatorship}
    The dictatorship of the most-demanding member is a dominant voting mechanism.
\end{thm}

Theorem \ref{thm:optimaldictatorship} rests on two fundamental insights:
\begin{enumerate}
    \item The $\crownset{m}$-dictatorship dominates any voting mechanism that is at least as likely to enact the policy as itself.
    \item Any voting mechanism is at least as likely to enact the policy as the $\crownset{m}$-dictatorship.
\end{enumerate}

The first is easy to see. The $\crownset{m}$-dictatorship clearly minimises committee members' information costs: they incur none. While doing so, it also ensures that a good policy is almost surely enacted. 
So, a voting mechanism which has a greater probability of enacting the policy than the $\crownset{m}$-dictatorship must enact bad policies more often. Such a mechanism is clearly dominated by the $\crownset{m}$-dictatorship. 

To see the second, suppose the lobbyist ignores the committee's intricate voting mechanism and conducts the experiment tailored for the $\crownset{m}$-dictatorship. I show that, with this experiment, the committee enacts the policy no less often than the $\crownset{m}$-dictatorship would. Indeed, the $\crownset{m}$-dictatorship enacts the policy if and only if this experiment generates the interim belief $\crownset{q}$. I show that this interim belief persuades committee members to vote for enactment \textit{regardless} of the voting mechanism they use. This is because $\crownset{q}$ (weakly) exceeds any member’s upper threshold of persuasion; if a member were deciding alone, she would enact the policy without acquiring further information. Whenever she believes her peers have no private information relevant to the decision anyway, she wishes to vote the same way.

\begin{proof}[Proof, Theorem \ref{thm:optimaldictatorship}:]
    Fix $\left( \left\{ V_j \right\}_{j=1}^{N}, \left\{ \mathcal{V}_j \right\}_{j=1}^{N}, D \right)$ to be an arbitrary voting mechanism, and let $\crownset{D}$ denote the voting rule for the $\crownset{m}$-dictatorship. 
    
    I will use Claim \ref{thm:1.1} to establish Theorem \ref{thm:optimaldictatorship}:

    \begin{claim}
        \label{thm:1.1}
        Any voting mechanism induces an equilibrium where, in the subgames with interim beliefs $r \geq \crownset{q}$, committee members:
        \begin{enumerate}[label=(\alph*)]
            \item acquire no information: their information strategies $\left\{ \mu_j^* (\cdot \mid r) \right\}_{r \in [0,1]}$ are such that for all $r \geq \crownset{q}$, $\mu_j^* (\cdot \mid r) = \delta_r$.
            \item vote to enact the policy: their voting strategies $\left\{ \tau_j^* (\cdot \mid r,q) \right\}_{(r,q) \in [0,1]^2}$ are such that for all $r \geq \crownset{q}$, $\tau_j^* (\cdot \mid r,r) = \delta_{v_j^*}$ where $D(v_1^*,...,v_N^*) = 1$.
        \end{enumerate}
    \end{claim}

    \begin{proof}
        Clearly, no strategy profile can leave the lobbyist strictly better off in the relevant subgames. So, it suffices to show that the asserted strategies are a mutual best response for all committee members. 

        Let us first establish part (b).
        Suppose every member except $m_i$ employs the asserted strategies. Define $\pi_i ( r,q )$ to be the payoff $m_i$ expects under this voting mechanism, given her interim belief $r$, posterior belief $q$, others' strategies, and her voting strategy. When $r \geq \crownset{q}$, other committee members acquire no information; so, their votes are not informative about the state. As such, for $r \geq \crownset{q}$:
        \begin{equation*}
            \pi_i (r,q) = \left[ q \times (1 + u_i) - 1 \right]  \times \int\limits_{V_i} D(v, \mathbf{v_{-i}^*}) \tau_i (dv)
        \end{equation*}
        given $m_i$'s chosen voting strategy $\left\{ \tau_i (\cdot \mid r,q) \right\}_{(r,q) \in [0,1]^2}$ and her peers' equilibrium votes, $\mathbf{v_{-i}^*}$.
        When $q_i$ is weakly above (strictly below) $(1 + u_i)^{-1}$,  $\pi_{i} (r,q)$ is strictly increasing (decreasing) in the enactment probability. So, $m_i$ must vote to maximise (minimise) this probability. 
        Since $\crownset{q} \geq \Bar{q}_i^* \geq (1 + u_i)^{-1}$, this implies that $m_i$ cannot profitably deviate from the asserted voting strategies. This proves part (b) of the claim.

        Now, let us prove part (a). The argument above established that, in equilibrium:
        \begin{equation*}
            \pi_i (r,q) =
                \begin{cases}
                    q \times (1 + u_i) - 1 & 
                    q \geq (1 + u_i)^{-1} \\
                    \min\limits_{v_i \in V_i} D(v_i, \mathbf{v_{-i}^*}) \times \left[ q \times (1 + u_i) - 1 \right] & 
                    q < (1 + u_i)^{-1} \\
                \end{cases}
        \end{equation*}
        for any $r \geq \crownset{q}$.
        To prove that acquiring no information is a best response for $m_i$, we must show:
        \begin{align}
            \label{eqn:1}
            &\int\limits_{[0,1]} \left[ \pi_i (q;r) - \pi_i (r;r) \right] d \mu (q) - C_i \left( \mu \right) \leq 0 & &\text{for all } r \geq \crownset{q} \text{ and } \mu \in \Delta_r \left( [0,1] \right)
        \end{align}

        To this end, define $\crownset{\pi}_i (q)$ to be the payoff $m_i$ expects from her preferred action at a posterior belief $q$; $\crownset{\pi}_i (q) := \max \left\{ 0, q \times (1 + u_i) - 1 \right\}$.
        Since $\Bar{q_i}^* \geq \crownset{q}$, $m_i$ acquires no information for interim beliefs $r \geq \crownset{q}$ under  her dictatorship. This implies:
        \begin{align}
            \label{eqn:2}
            &\int\limits_{[0,1]} \left[ \crownset{\pi}_i(q) - \crownset{\pi}_i(r) \right]  d \mu (q) - C_i \left( \mu \right) \leq 0 & &\text{for all } r \geq \crownset{q} \text{ and } \mu \in \Delta_r \left( [0,1] \right)
        \end{align} 
        Now let $\Tilde{\pi}_i (q,r) := \left[ \pi_i (q,r) - \pi_i (r,r) \right] - [ \crownset{\pi}_i (q) - \crownset{\pi}_i (r) ]$. \ref{eqn:1} follows from \ref{eqn:2} once we prove:
        \begin{align}
            \label{eqn:3}
            &\int\limits_{[0,1]} \Tilde{\pi}_i (q,r)  d \mu (q) \leq 0 & &\text{for all }
             r \geq \crownset{q} \text{ and }
            \mu \in \Delta_r \left( [0,1] \right)
        \end{align}
        %
        %\ref{eqn:3} reflects that a committee member finds information most valuable under her dictatorship.

        To see \ref{eqn:3}, first consider the behaviour of $\Tilde{\pi}_i (q,r)$ as a function of $q$, for some fixed $r$. This function is piecewise linear, has a kink at $q = (1 + u_i)^{-1}$, increases before this kink, and decreases after it. Therefore, for any fixed $r$, $\Tilde{\pi}_i (q,r)$ is a concave function of $q$. Moreover, $\Tilde{\pi} (r,r) = 0$, trivially. By Jensen's Inequality, we then obtain:
        \begin{align*}
            &\int\limits_{[0,1]} \Tilde{\pi}_i (q;r)  d \mu (q) \leq \Tilde{\pi}_i (r;r) = 0 &\text{for all } \mu \in \Delta_r \left( [0,1] \right)
        \end{align*}
    \end{proof}

    In the equilibrium the $\crownset{m}$-dictatorship induces, no committee member expects to incur information acquisition costs. Moreover, a good policy is almost surely enacted. Therefore, to prove Theorem \ref{thm:optimaldictatorship}, it suffices to show that any other mechanism has a higher equilibrium probability of enacting a bad policy. 

    In the equilibrium of the $\crownset{m}$-dictatorship, the policy is enacted if and only if the lobbyist's experiment induces the interim belief $\crownset{q}$. Claim \ref{thm:1.1} establishes that this interim belief leads to enactment under any voting mechanism. So, by performing the same experiment, the lobbyist can secure a weakly greater probability of enactment under any such mechanism. This implies that any such mechanism is likelier to enact the policy than the $\crownset{m}$-dictatorship, in equilibrium. Since the $\crownset{m}$-dictatorship never blocks a bad policy, any such mechanism must be likelier to enact a bad policy, too.
\end{proof}

Theorem \ref{thm:optimaldictatorship} aligns with the findings of \textcite{alonso_camara_2016} and \textcite{bardhi_guo_2018}. These papers also feature a lobbyist\footnote{\textit{The politician} in the former, \textit{the Sender} in the latter.} who wishes to persuade the committee to enact a policy by designing an experiment about its unknown benefit. However, their committee members have no access to private information. They find that, when committee members' preferences mirror those considered in this paper, every committee member prefers \textit{unanimity rule} to other $k$-majority rules---mechanisms which enact the policy if at least $k$ members vote ``yes''. This is because, under unanimity, the most-demanding member\footnote{
    In their setting, committee members' upper thresholds of persuasion coincide with their thresholds of indifference, since they cannot acquire any private information. Then, the most-demanding member is simply the one who is least biased towards the policy; i.e., with the lowest $u_i$. 
} becomes the \textit{de facto} dictator: others never vote ``no'' when she votes ``yes'', and when she votes ``no'', the policy is blocked. To persuade her, the lobbyist provides the committee more information than he would under any other $k$-majority rule. As such, the committee makes the best decisions under unanimity.

Proposition \ref{prop:dict_mostinfo} establishes an analogous result for our setting.

\begin{prop}
    \label{prop:dict_mostinfo}
    Let $\varepsilon$ and $\crownset{\varepsilon}$ be two Blackwell experiments inducing the distributions $\mu$ and $\crownset{\mu}$ over interim beliefs, where $\supp \crownset{\mu} = \{ 0, \crownset{q} \}$. If $\varepsilon$ is Blackwell more informative than $\crownset{\varepsilon}$, the lobbyist expects a lower equilibrium payoff with $\mu$ than with $\crownset{\mu}$ under any voting mechanism. 
\end{prop}

\begin{proof}
    By Blackwell's Theorem (\textcite{blackwell_53}),\footnote{
        The reader may note that \textcite{blackwell_53} states the equivalence only for the standard measure of an experiment; i.e., the distribution of interim beliefs under the prior $p = 1/2$, here. However, it is easily confirmed that \textcite{blackwell_53}'s proof for the ``only if'' direction holds for any the interim belief distribution generated by \textit{any} prior $p \in [0,1]$.
    } if $\varepsilon$ is Blackwell more informative than $\crownset{\varepsilon}$, then for any convex function ${\phi}$ over $\mathbb{R}$:
    \begin{equation*}
        \int \phi(r) d\mu(r) \geq \int \phi(r) d \crownset{\mu}(r)
    \end{equation*}
    By Theorem 8 in \textcite{strassen_1965}, this condition implies the existence of a belief martingale $\{ \xi_1, \xi_2 \}$ such that ${\xi_1} \sim \crownset{\mu}$ and ${\xi_2} \sim \mu$. Since $\{ \xi_1, \xi_2 \}$ is a belief martingale, $\Prob \left( \xi_2 = 0\mid \xi_1 = 0 \right) = 1$. Then, $\xi_2 \neq \xi_1$ only when $\xi_1 = \crownset{q}$. But by Claim \ref{thm:1.1}, the interim belief $\crownset{q}$ leads to the enactment of the policy under any voting mechanism. 
    So, the interim belief martingale $\xi_2$ can never increase the lobbyist's expected payoff above that induced by $\xi_1$.

\end{proof}

Proposition \ref{prop:dict_mostinfo}, however, only provides a partial intuition for Theorem \ref{thm:optimaldictatorship}. Even though other mechanisms extract less information from the lobbyist than the $\crownset{m}$-dictatorship does, they may incentivise committee members to acquire more information. By doing so, they may help the committee accumulate more information overall, leading to better decisions. 

What drives Theorem \ref{thm:optimaldictatorship} is that the lobbyist \textit{always} prevents this. He can always do so with the same experiment he would perform for a $\crownset{m}$-dictatorship. That experiment discourages committee members from acquiring further information and persuades them to enact the policy as frequently as the $\crownset{m}$-dictatorship would, regardless of the voting mechanism they use. He always wishes to do so, too. A $\crownset{m}$-dictatorship never blocks a good policy; so, the committee can improve its decisions only by blocking bad policies more often. So, the lobbyist never wishes the committee to take better decisions than the $\crownset{m}$-dictatorship would.

\section{Discussion}
\label{section:discussion}

In this Section, I conclude by discussing the robustness of Theorem \ref{thm:optimaldictatorship} to various modelling assumptions. 

\subsection*{Simplifying Assumptions}

The model assumes that committee members' experiments are mutually independent conditional on the state. This means that committee members cannot directly gather information on their peers' private signals; they can only do so indirectly by gathering information about the unknown state that generates those signals. It also assumes that the lobbyist can only provide \textit{public} information---send the same information to all committee members. It does not explicitly cover the case of \textit{private} information provision, where the lobbyist may be allowed to communicate different bits of information to each committee member. 

While both restrictions streamline exposition, neither affects the results presented in the paper. Clearly, neither restriction is relevant for how dictatorships behave; hence, for Lemma \ref{lem:dictatoreqm}. Theorem \ref{thm:optimaldictatorship} and Proposition \ref{prop:dict_mostinfo} rest on the key insight that the information that persuades the $\crownset{m}$-dictatorship persuades any other committee, too: it discourages any such committee from acquiring further information, and ensures that it enacts the policy at least as often as the $\crownset{m}$-dictatorship. This insight hinges only on committee members' cost of acquiring information about the unknown state $\state$ and their preferences about the policy; it is not sensitive to the removal of the aforementioned restrictions on the model. Thus, the results in this paper can easily be extended to a setting where committee members can acquire information about each others' signals, and the lobbyist can communicate privately with each member.

\subsection*{Equilibrium Selection}

I focus on the lobbyist-preferred subgame perfect equilibria that different voting mechanisms induce. Under some voting rules---for instance, simple majority---this may lead to odd voting strategies. Particularly, committee members may vote to enact the policy regardless of their beliefs and preferences, simply because each believes that her vote will not be pivotal. To avoid such strategies, one may refine the set of equilibria further; for instance, by excluding equilibria in weakly dominated voting strategies (e.g., like \textcite{alonso_camara_2016} and \textcite{heese_lauermann_2024}). 

Such refinements on voting strategies have no bearing on the results presented here. They do not affect the equilibria dictatorships induce: dictators always vote for their preferred outcome. They do not affect the key insight behind Theorem \ref{thm:optimaldictatorship} or Proposition \ref{prop:dict_mostinfo} either: at the interim belief $\crownset{q}$ that persuades $\crownset{m}$ to enact the policy, any committee member prefers to see the policy enacted rather than blocked. So, information that persuades the $\crownset{m}$-dictatorship continues to persuade any other committee.

Nonetheless, the information strategies my equilibrium concept selects do play a key role in establishing Theorem \ref{thm:optimaldictatorship} and Proposition \ref{prop:dict_mostinfo}. Following the interim belief $\crownset{q}$, no committee member acquires additional information---she instead votes to enact the policy. This information strategy is a best response given her belief that, in equilibrium, her peers will not acquire any private information either. Since these strategies secure the lobbyist's favourite outcome, they constitute the lobbyist-preferred subgame perfect equilibrium.

This, however, does not preclude subgame perfect equilibria where the committee \textit{does} acquire private information following the interim belief $\crownset{q}$. Indeed, consider a committee member who instead believes that her peers {will} acquire additional private information at this interim belief. She no longer believes that the committee will decide based on the interim belief $\crownset{q}$; her peers may grow less certain that the policy was good, and vote against it. If they hold such private information, should she vote to allow her peers to block the policy? Or, should she counteract them, possibly reassured by her positive private information? To respond to this contingency, she may now find it useful to acquire private information herself.

This incentive to acquire information in order to best respond to others' renewed uncertainty mirrors the \textit{encouragement effect}, first explored by \textcite{bolton_harris_1999}. Unlike their setting, committee members must \textit{simultaneously} decide whether to acquire information, \textit{before} they discover what their peers know. Nonetheless, an analogue of the encouragement effect may still emerge, thanks to members' incentives to best respond to their peers' informed votes.

My equilibrium selection rule precludes such equilibria---the lobbyist is best off when committee members vote in favour of the policy without acquiring further information. Studying when these equilibria emerge and the outcomes they lead to is itself a very interesting and promising line of inquiry, left to future research.

\subsection*{The Lobbyist's Information Production Costs}

The model assumes that the lobbyist can produce information freely---he can \textit{costlessly} choose \textit{any} Blackwell experiment about the state. This assumption has become standard and widespread in the literature following \textcite{kg2011}'s seminal paper. Nonetheless, it is a very powerful assumption, and a key driver of the results presented in this paper. To illustrate how it shapes the results, I extend the framework this paper analysed by constraining the lobbyist's ability to produce information for the committee. 

Let $E$ be a finite set of Blackwell experiments about the unknown state $\state$, and $\Delta_{E} \subset \Delta_p ([0,1])$ be the set of interim belief distributions some experiment in $E$ may generate. Suppose that the lobbyist's information strategy $\mu_l$ must now be an element of this set $\Delta_{E}$, rather than any Bayes plausible distribution $\mu \in \Delta_p ([0,1])$. Otherwise, assume that the setting remains the same.

\begin{prop}
    \label{prop:constrainedlobbyist}
    Suppose that $\crownset{\mu} \in \Delta_E$, where $\supp \crownset{\mu} = \{ 0, \crownset{q} \}$. Then, Lemma \ref{lem:dictatoreqm}, Theorem \ref{thm:optimaldictatorship} and Proposition \ref{prop:dict_mostinfo} remain valid. 
\end{prop}

\begin{proof}
    Proposition \ref{prop:dict_mostinfo} trivially remains valid, as it does not depend on the set of experiments available to the lobbyist.

    Lemma \ref{lem:dictatoreqm} establishes that under the $\crownset{m}$-dictatorship, no strategy $\mu_l \in \Delta_{p} ([0,1])$ delivers a greater equilibrium expected payoff than $\crownset{\mu}$ for the lobbyist. Therefore, $\crownset{\mu}$ remains optimal for the lobbyist when he is restricted to strategies $\mu \in \Delta_E$. As such, Lemma \ref{lem:dictatoreqm}'s characterisation of the equilibrium under the $\crownset{m}$-dictatorship remains valid. Claim \ref{thm:1.1} and the ensuing argument used to establish Theorem \ref{thm:optimaldictatorship} only rely on Lemma \ref{lem:dictatoreqm}'s characterisation of the equilibrium under a $\crownset{m}$-dictatorship, the lobbyist's ability to employ the strategy $\crownset{\mu}$, his preferences regarding policy outcomes, and committee members' best responses. Therefore, Theorem \ref{thm:optimaldictatorship} also remains valid.  
    
\end{proof}

When $\crownset{\mu} \notin E$, however, Proposition \ref{prop:constrainedlobbyist} fails. I show this through a simple counter-example. 

Suppose there are two committee members, $m_1$ and $m_2$, who share the prior belief $p = 0.5$ with the lobbyist. $m_1$ and $m_2$ receive the payoffs $u_1 = 1/4$ and $u_2 = 2/3$ when a good policy is enacted. Thus, their thresholds of indifference are $(1 + u_1)^{-1} = 0.8$ and $(1 + u_2)^{-1} = 0.6$, respectively. Define $\pi_i(q)$ to be the payoff member $m_i$ expects when her preferred decision is taken at posterior belief $q$:
\begin{align*}
    \pi_1(q) &:=
    \begin{cases}
        0 & q < 0.8 \\
        5/4 \times (q - 0.8) & q \geq 0.8
    \end{cases}
    &
    \pi_2(q) &:=
    \begin{cases}
        0 & q < 0.6 \\
        5/3 \times (q - 0.6) & q \geq 0.6
    \end{cases}
\end{align*}
Let $m_i$'s cost of acquiring information be defined by the convex function $c_i (q) := q^2 + \pi_i (q)$. 

In turn, suppose that the lobbyist's information strategy must be an element of $\Delta_E = \{ \mu_1, \mu_2 \}$, where $\mu_1$ and $\mu_2$ are the unique distributions such that $\mu_1, \mu_2 \in \Delta_{0.5} ([0,1])$, $\supp \mu_1 = \{ 0.2, 0.8 \}$, and $\supp \mu_2 = \{ 0, 0.6 \}$. 

In this counter example, $m_1$ is the most-demanding member of the committee: her upper threshold of persuasion equals $\Bar{q}_1^* = 0.8$ against $m_2$'s $\Bar{q}_2^* = 0.6$. Yet, the $m_1$-dictatorship does not dominate the $m_2$-dictatorship. 

To see this, first note that neither member finds it optimal to acquire any information under her dictatorship; both members have the equilibrium information strategy $\{ \mu_i^* (\cdot \mid r) : \mu_i^* ( \cdot \mid r ) = \delta_r \text{ for all } r \in [0,1] \}$. This is because $m_i$'s equilibrium information strategy must solve the following problem:
\begin{equation*}
    \mu_i^* ( \cdot \mid r ) \in \argmax\limits_{\mu \in \Delta_r ([0,1])} \int\limits_{[0,1]} \left[ \pi_i (q) - c_i(q) \right] d \mu (q) - [\pi_i (r) - c_i (r)] 
\end{equation*}
The function $\pi_i (q) - c_i (q) = - q^2$ is strictly concave throughout its domain. Due to Jensen's Inequality, this implies that $\mu_i( \cdot \mid r ) = \delta_r$ is the unique solution to this problem.

Thus, both members' upper thresholds of persuasion coincide with their thresholds of indifference: $ \Bar{q}_i^* = (1 + u_i)^{-1} $. Under the $m_1$-dictatorship, the strategy $\mu_2$ never allows the lobbyist to enact the policy---its support lies strictly below $m_1$'s upper threshold of indifference. In contrast, the strategy $\mu_1$ enacts the policy when the interim belief $r = 0.8$ is realised. Hence, $\mu_1$ is the lobbyist's equilibrium strategy under $m_1$'s dictatorship. This dictatorship achieves the following conditional probabilities:
\begin{align*}
    \Prob_{\mu_1} ( \text{enact} \mid \state = \state_1 ) &= \Prob_{\mu_1} ( r = 0.8 \mid \state = \state_1 ) = 4/5
    \\
    \Prob_{\mu_1} ( \text{block} \mid \state = \state_0 ) &= \Prob_{\mu_1} ( r = 0.2 \mid \state = \state_0 ) = 4/5
\end{align*}

In contrast, $\mu_2$ is the optimal strategy under the $m_2$-dictatorship. This strategy enacts the policy when the interim belief $r = 0.6$ is realised. This occurs with probability $5/6$, as one may easily verify using Bayes Rule. In contrast, $\mu_1$ enacts the policy only when the interim belief $r = 0.8$ is realised. This occurs with probability $1/2$. Hence, $\mu_2$ is the lobbyist's equilibrium strategy under $m_2$'s dictatorship. This dictatorship achieves the following conditional probabilities:

\begin{align*}
    \Prob_{\mu_2} ( \text{enact} \mid \state = \state_1 ) &= \Prob_{\mu_2} ( r = 0.6 \mid \state = \state_1 ) = 1
    \\
    \Prob_{\mu_2} ( \text{block} \mid \state = \state_0 ) &= \Prob_{\mu_2} ( r = 0 \mid \state = \state_0 ) = 1/3
\end{align*}

This verifies that $m_1$'s dictatorship does not dominate $m_2$'s dictatorship, and vice versa. Under $m_2$'s dictatorship, the lobbyist faces a lower upper threshold of persuasion than under $m_1$'s dictatorship. As such, he tailors the information he provides to $m_2$ so that she is likelier to enact the policy than $m_1$ is. Particularly, $m_2$ becomes likelier to enact a bad policy than $m_1$ is. However, she also becomes likelier to enact a good policy. Such a trade-off never emerges if the lobbyist can generate \textit{any} interim belief distribution $\mu \in \Delta_p ([0,1])$. He then ensures that dictators always enact a good policy.

\newpage

\printbibliography

\end{document}